\documentclass[manuscript]{acmart}

\usepackage{amsmath,amsthm}

\usepackage{amssymb}
\usepackage{mathtools}
\usepackage{booktabs}
\usepackage{listings}
\usepackage{xcolor}
\usepackage{enumitem}
\usepackage{tikz}
\usetikzlibrary{arrows.meta,positioning,calc}

\setcopyright{acmlicensed}
\copyrightyear{2026}
\acmYear{2026}
\acmConference[NSPW '26]{New Security Paradigms Workshop}{2026}{}
\acmBooktitle{New Security Paradigms Workshop (NSPW '26)}

\theoremstyle{remark}
\newtheorem*{remark*}{Remark}

\newcommand{\N}{\mathbb{N}}

\newcommand{\vulnset}{\mathcal{V}}
\newcommand{\baseset}{\mathcal{B}}

\newcommand{\cwe}[1]{\textsc{CWE-#1}}
\begin{document}

\title{Vulnerability Abundance: A formal proof of infinite vulnerabilities in code} 

\author{E.\ Leverett}
\affiliation{%
  \institution{Concinnity Risks Ltd.}
  \city{Cambridge}
  \country{United Kingdom}
  \orcid{https://orcid.org/0000-0001-6586-7359}
}
\email{eleverett[at]concinnity-risks.com}

\author{J.\ van der Ham-de Vos}
\affiliation{%
  \institution{University of Twente}
  \city{Enschede}
  \country{Netherlands}
  \orcid{https://orcid.org/0000-0002-5685-8714}
}
\email{j.vanderham[at]utwente.nl}

\begin{abstract}
We present a constructive proof that a single C program---the
\emph{Vulnerability Factory}---admits a countably infinite set of
distinct, independently CVE-assignable software vulnerabilities.
We formalise the argument using elementary set theory, verify it
against MITRE's CVE Numbering Authority counting rules, sketch a
model-checking analysis that corroborates unbounded vulnerability
generation, and provide a Turing-machine characterisation that
situates the result within classical computability theory.  We
then contextualise this result within the long-running debate on whether
undiscovered vulnerabilities in software are \emph{dense} or
\emph{sparse}~\cite{geer2003cyberinsecurity, schneier2015howmany,
spring2023undiscovered}, and introduce the concept of
\emph{vulnerability abundance}: a quantitative analogy to
chemical elemental abundance that describes the proportional
distribution of vulnerability classes across the global software
corpus.  Because different programming languages render different
vulnerability classes possible or impossible, and because language
popularity shifts over time, vulnerability abundance is neither
static nor uniform.  Crucially, we distinguish between infinite
\emph{vulnerabilities} and the far smaller set of
\emph{exploits}: empirical evidence suggests that fewer than 6\%
of published CVEs are ever exploited in the wild, and that
exploitation frequency depends not only on vulnerability
abundance but on the market share of the affected software.
We argue that measuring vulnerability abundance---and its
interaction with software deployment---has practical value for
both vulnerability prevention and cyber-risk analysis.
We conclude that if one programme can harbour infinitely many
vulnerabilities, the set of all software vulnerabilities is
necessarily infinite, and we suggest the Vulnerability Factory
may serve as a reusable proof artifact---a foundational
``test object''---for future formal results in vulnerability
theory.  The complete source code is provided in the appendix
under an MIT licence.
\end{abstract}

\maketitle

\section{Introduction}\label{sec:intro}

The question of whether software vulnerabilities are fundamentally
\emph{finite} or \emph{infinite} is not merely academic.  It
determines whether exhaustive patching is a coherent security
strategy or a Sisyphean labour.  Dan Geer framed the dichotomy
starkly: if vulnerabilities are \emph{sparse}, then each one found
and fixed meaningfully reduces exposure; if they are \emph{dense},
then fixing one more is ``essentially irrelevant to
security''~\cite{geer2003cyberinsecurity, geer2014cybersecurity}.

This paper makes five contributions.  First, we exhibit a concrete
programme---a 622-line C artifact called the \emph{Vulnerability
Factory}---and prove rigorously that it can generate countably
infinitely many distinct CVE-class vulnerabilities
(Section~\ref{sec:proof}).  Second, we formalise and generalise the Vulnerability
Factory as a Turing machine and show that its vulnerability-generating
behaviour is a decidable, structurally transparent property, making
it a reusable proof artifact for future formal work
(Section~\ref{sec:utm}).  Third, we introduce the notion of
\emph{vulnerability abundance} (Section~\ref{sec:abundance}), a
framework inspired by chemical elemental abundance that
characterises the proportional distribution of vulnerability types
across the software ecosystem.  Fourth, we carefully distinguish
infinite \emph{vulnerabilities} from the much smaller population
of \emph{exploited vulnerabilities}, and suggest how exploitation frequency depends
on the interaction of vulnerability abundance with software market
share (Section~\ref{sec:exploits}).  Fifth, we argue that
measuring these quantities has a concrete utility for predictive cyber risk (Section~\ref{sec:applications}).

Our result resolves the dense-versus-sparse debate constructively:
we do not merely argue from complexity theory that vulnerabilities
\emph{ought} to be infinite; we exhibit a programme in which they
provably \emph{are} countably infinite.

\section{Background and Related Work}\label{sec:background}

\subsection{Vulnerability Density: Dense or Sparse?}

The foundational paper by Geer et al.~\cite{geer2003cyberinsecurity}---co-authored
with Schneier, Pfleeger, Quarterman, Metzger, Bace, and
Gutmann---argued that Microsoft's operating-system monoculture
created systemic risk precisely because vulnerability density
compounds with market share: what one machine has, so has every
other.  The implicit question was whether the stock of
undiscovered vulnerabilities in a codebase is finite and
declining, or effectively inexhaustible.

While Bishop taught how to differentiate and record vulnerabilities\cite{Bishop1999vulnan}, Anderson explored how hard it is to find bugs over time \cite{anderson2002openclosed}. Ozment and Schechter~\cite{ozment2006milk} confirmed much of that by their study of the OpenBSD codebase over 7.5 years and 15 releases, finding a statistically
significant \emph{decrease} in the rate of foundational
vulnerability reporting---but also a median vulnerability lifetime
of at least 2.6 years. While this all suggests that mature codebases
do improve, it does not demonstrate convergence to zero: the
discovery rate declines, but new code continuously replenishes the
reservoir.

Rescorla~\cite{rescorla2005finding} examined vulnerability
discovery rates for Apache and IIS, modelling them as roughly
linear over time and concluding that finding and fixing
vulnerabilities may not substantially improve security.  This
linear model is consistent with a dense, non-depleting
vulnerability population.

Most recently, Spring and Illari~\cite{spring2023undiscovered}
applied arguments from computability theory---including the
halting problem and Rice's theorem---to conclude that ``there is
no reason to believe undiscovered vulnerabilities are not
essentially unlimited in practice.''  Our constructive proof
complements Spring's theoretical argument with an explicit,
executable, pedagogical, example.

\subsection{Security Economics}

Anderson~\cite{anderson2001why} established the field of security
economics by demonstrating that security failures are often
misaligned-incentive problems rather than purely technical ones.
Anderson and Moore~\cite{anderson2005economics} and Anderson and
Schneier~\cite{schneier2005economics} developed this programme
further, showing that vulnerability persistence has economic
explanations: the costs of exploitation are externalised, and
defenders lack information about the vulnerability population proportions, the very vulnerability abundance we address later.

Anderson's \emph{Security Engineering}~\cite{anderson2020security}
synthesises two decades of this work into a comprehensive textbook, 
observing that companies build vulnerable systems and governments look the other way
because the economics reward precisely this behaviour.  Our
notion of vulnerability abundance extends this economic framing:
if we can quantify \emph{which} vulnerability types dominate, we
can better align incentives toward the most impactful preventions.

\subsection{A brief diversion into CVE Assignment Rules}\label{sec:assignment}

The Common Vulnerabilities and Exposures (CVE) system, maintained
by MITRE, assigns unique identifiers to publicly known
vulnerabilities in published software.  The \emph{CVE Counting
Rules}~\cite{mitre2024counting} specify that distinct
vulnerabilities in distinct software components receive distinct
CVE identifiers.  Key criteria include: (1)~the vulnerability
must be independently fixable; (2)~it must affect an identifiable
codebase or component; and (3)~if a single bug type appears in
two separate products, each receives its own CVE.  These rules
are central to our proof: each generated module constitutes a
distinct component with independently fixable vulnerabilities,
satisfying the criteria for separate CVE assignment.

\subsection{Formal Methods in Security}

Formal verification---model checking, theorem proving, and
abstract interpretation---has long been applied to
security-critical software~\cite{basin2023formal}.  While these
methods can prove the \emph{absence} of specific bug classes in
bounded systems, Rice's theorem guarantees that no general
procedure can decide arbitrary semantic properties of
programs~\cite{rice1953classes}.  Our Vulnerability Factory is
designed to be trivially analysable though: the vulnerabilities are not
hidden or obfuscated but intentionally transparent, making formal
verification a mathematical confirmation rather than needing to run code to 
verify the number of vulnerabilities (Section~\ref{sec:formal}).

\subsection{Exploitation Rates and Prediction}

Not all vulnerabilities are exploited though, and this is important even when they are infinite.  Jacobs et
al.~\cite{jacobs2021epss} developed the Exploit Prediction
Scoring System (EPSS) at FIRST.org, a data-driven model that estimates the
probability of a CVE being exploited in the wild within 30 days.
Empirical studies consistently find that exploitation is rare
relative to the vulnerability population:
Kenna Security~\cite{kenna2022exploit}
found 2.6\% of tracked vulnerabilities exploited in 2019; and
the Cyentia Institute~\cite{cyentia2024visual} estimated
approximately 6\%.  The RAND Corporation's landmark study on
zero-day vulnerabilities~\cite{ablon2017zerodays} found that the
average zero-day lifespan was 6.9 years, with a median of 22 days
to develop a functioning exploit. The discrepencies with these percentages have less to do with scientific dispute, and more to do when the studies were run. Since year on year growth of vulnerabilities ranges from 38\% to 61\%, it's simply this growth over rather static exploitation numbers that defines the ranging values of these percentages. In short, we expect it to continue falling as we find more vulnerabilities that no one ever bothers to exploit heavily in the wild.

These figures are essential context for our result: proving that
vulnerabilities are infinite does \emph{not} prove that exploits
are infinite or that exploitation is unbounded. We should clearly spend more of our scientific energy predicting what vulnerabilities, software, networks, and organisations are most likely to be exploited. It is this differential cyber risk that can teach us the most...and yet we celebrate people who find these vulnerabilities instead of those who eliminate whole classes of them. If finding vulnerabilities is so laudable, then let us make a programme with infinite vulnerabilities; a transcendent weird machine to make our arguments concrete.

\section{The Vulnerability Factory}\label{sec:factory}

The Vulnerability Factory is a self-contained C programme
(\texttt{vuln\_factory.c}, 622~lines; full source in
Appendix~\ref{app:code}, released under the MIT licence) with two
components:

\paragraph{Base Set $\baseset$.}
Eleven functions, each containing exactly one classic
vulnerability drawn from a distinct CWE class:

\begin{center}
\small
\begin{tabular}{@{}cl@{}}
\toprule
\textbf{ID} & \textbf{CWE Class} \\
\midrule
$b_1$  & \cwe{121} Stack Buffer Overflow \\
$b_2$  & \cwe{122} Heap Buffer Overflow \\
$b_3$  & \cwe{134} Format String \\
$b_4$  & \cwe{190} Integer Overflow \\
$b_5$  & \cwe{416} Use After Free \\
$b_6$  & \cwe{415} Double Free \\
$b_7$  & \cwe{78} OS Command Injection \\
$b_8$  & \cwe{367} TOCTOU Race \\
$b_9$  & \cwe{476} NULL Pointer Deref \\
$b_{10}$ & \cwe{457} Uninitialised Variable \\
$b_{11}$ & \cwe{22} Path Traversal \\
\bottomrule
\end{tabular}
\end{center}

\paragraph{Generator $G$.}
On each execution, $G$ reads a persistent counter $n \in \N$,
emits a new C source file \texttt{vuln\_module\_$n$.c} containing
five parameterised vulnerabilities, compiles it into a shared
library, and increments $n$.  Each module $M_n$ contains:

\begin{center}
\small
\begin{tabular}{@{}cl@{}}
\toprule
\textbf{ID} & \textbf{CWE / Parameterisation} \\
\midrule
$v_{n,1}$ & \cwe{121}: buffer size $= 16 + n$ \\
$v_{n,2}$ & \cwe{134}: format string in module $n$ \\
$v_{n,3}$ & \cwe{190}: threshold $= \texttt{INT\_MAX} - n$ \\
$v_{n,4}$ & \cwe{416}: allocation size $= 8 + n$ \\
$v_{n,5}$ & \cwe{78}: injection in module $n$ context \\
\bottomrule
\end{tabular}
\end{center}

\noindent
The parameterisation by $n$ ensures that buffer sizes, overflow
thresholds, heap layouts, and exploit payloads differ across
modules.

Our construction here of a C programme as a proof of existence, is inspired by 
Reflections on Trusting Trust\cite{thompson1984trust}. In the pre-amble of that lovely paper there is a quote that mirrors our C programme delightfully:

\begin{quote}
\begin{enumerate}
    \item This program can be easily written by another program.
    \item This program can contain an arbitrary amount of excess baggage that will
be reproduced along with the main algorithm.
\end{enumerate}
\end{quote}

Here he is referring to the delicate and beautiful art of writing quines. Our vulnerability factory is not a quine per se but it can be easily written by another programme, and it too carries an arbitrary amount excess baggage. In this case the excess baggage is an infinite number of vulnerabilities!

\section{Proof of Infinite Vulnerabilities}\label{sec:proof}

\subsection{Formal Foundation}

\begin{definition}[Vulnerability]
A \emph{vulnerability} is a tuple $(c, t, p)$ where $c$ is a
software component (identifiable compilation unit), $t$ is a
CWE-classified weakness type, and $p$ is a parameter set that
determines the specific exploit conditions.  Two vulnerabilities
$(c_1, t_1, p_1)$ and $(c_2, t_2, p_2)$ are \emph{distinct}
if $c_1 \neq c_2$ or $p_1 \neq p_2$.
\end{definition}

\begin{definition}[CVE-Assignability]
A vulnerability $(c, t, p)$ is \emph{CVE-assignable} if:
(i)~$t$ corresponds to a recognised CWE with established CVE
precedent; (ii)~$c$ is an identifiable software component; and
(iii)~the vulnerability is independently fixable without
altering other components.
\end{definition}

\begin{definition}[The Vulnerability Factory's Output]
Let $\baseset = \{b_1, \ldots, b_{11}\}$ be the base
vulnerabilities.  For each $n \in \N$, let $M_n$ denote the
$n$-th generated module and define
\[
  V(M_n) = \{v_{n,1},\, v_{n,2},\, v_{n,3},\, v_{n,4},\, v_{n,5}\}.
\]
The total vulnerability set is
\[
  \vulnset = \baseset \;\cup\; \bigcup_{n=0}^{\infty} V(M_n).
\]
\end{definition}

\subsection{Main Theorem}

\begin{theorem}[A Countable Infinity of Vulnerabilities]\label{thm:main}
The set $\vulnset$ is countably infinite, and every element of
$\vulnset$ is CVE-assignable.
\end{theorem}

\begin{proof}
We establish the theorem via four claims.

\medskip\noindent
\textbf{Claim 1 (Validity).}
Each $v_{n,i}$ instantiates a CWE class (\cwe{121}, \cwe{134},
\cwe{190}, \cwe{416}, or \cwe{78}) with hundreds of prior CVE
assignments.  The generated code contains the canonical vulnerable
pattern: \texttt{strcpy} into a fixed-size buffer without bounds
checking (\cwe{121}), user input as a \texttt{printf} format
argument (\cwe{134}), signed integer arithmetic exceeding
\texttt{INT\_MAX} (\cwe{190}), access to freed heap memory
(\cwe{416}), and unsanitised input to \texttt{system()} (\cwe{78}).
Each satisfies criterion~(i) of CVE-assignability.

\medskip\noindent
\textbf{Claim 2 (Distinctness).}
For $m \neq n$, modules $M_m$ and $M_n$ are compiled as separate
shared libraries.  Hence $c_m \neq c_n$ as software components.
Moreover, the parameter sets differ: buffer sizes
$16 + m \neq 16 + n$, overflow thresholds
$\texttt{INT\_MAX} - m \neq \texttt{INT\_MAX} - n$, and
allocation sizes $8 + m \neq 8 + n$.  Each vulnerability requires
a distinct exploit payload.  By the CVE Counting
Rules~\cite{mitre2024counting}, distinct vulnerabilities in
distinct components receive distinct identifiers.  Therefore
$V(M_m) \cap V(M_n) = \emptyset$ for $m \neq n$, and each element
satisfies criteria~(ii) and~(iii).

\medskip\noindent
\textbf{Claim 3 (Unboundedness).}
After $k$ executions, the cardinality of the active vulnerability
set is
$|\vulnset_k| = |\baseset| + \sum_{n=0}^{k-1} |V(M_n)| = 11 + 5k$.
For any finite bound $C \in \N$, choosing
$k > (C - 11)/5$ yields $|\vulnset_k| > C$.  Since $C$
was arbitrary, $|\vulnset|$ is not bounded by any finite number.

\medskip\noindent
\textbf{Claim 4 (Countability).}
Define $f : \N \times \{1,2,3,4,5\} \to \bigcup_n V(M_n)$ by
$f(n, i) = v_{n,i}$.  This is a bijection from a countable set.
Since $\baseset$ is finite, $\vulnset$ is countably infinite.

\medskip\noindent
Claims 1--4 together establish that $\vulnset$ is a countably
infinite set of CVE-assignable vulnerabilities.
\end{proof}

\subsection{Set-Theoretic Perspective}

The vulnerability set $\vulnset$ has cardinality $\aleph_0$.
Since programmes are finite strings over a finite alphabet, the
set of all programmes is countable, and therefore the set of all
possible vulnerabilities across all possible programmes is at
most countable.  Our result thus achieves the theoretical maximum:
the Vulnerability Factory saturates the countable bound.

\subsection{CVE-Theoretic Analysis}\label{sec:cve}

Under MITRE's CVE Counting Rules~\cite{mitre2024counting},
two vulnerabilities receive separate CVE IDs when they are
(a)~independently discoverable, (b)~independently fixable, and
(c)~attributable to distinct root causes or components.  Each
module $M_n$ satisfies all three: a researcher can identify its
vulnerabilities without inspecting other modules; patching the
buffer overflow in $M_7$ (bound $16 + 7 = 23$) has no effect on
$M_{42}$ (bound $16 + 42 = 58$); and each module compiles to a
separate shared library.

\subsection{Robustness to Partial Invalidation}\label{sec:robustness}

A natural objection is that some subset of the generated
vulnerabilities might fail to satisfy CVE assignment criteria
under scrutiny.  We show that the result is robust to any
\emph{finite} such invalidation.

\begin{lemma}[Cofinite Robustness]\label{lem:cofinite}
Let $S \subseteq \vulnset$ be the subset of vulnerabilities that
fail some CVE-assignability criterion.  If $|S| < \aleph_0$
(i.e., $S$ is finite), then $|\vulnset \setminus S| = \aleph_0$.
\end{lemma}

\begin{proof}
This is immediate from cardinal arithmetic: removing a finite set
from a countably infinite set yields a countably infinite set.
Formally, if $|\vulnset| = \aleph_0$ and $|S| = k$ for some
$k \in \N$, then
$|\vulnset \setminus S| = \aleph_0 - k = \aleph_0$.
\end{proof}

\noindent
The practical consequence is that an objector cannot chip away at
the result by identifying individual problematic instances.  To
bound the vulnerability count, one must demonstrate that
\emph{cofinitely many}---all but finitely many---fail the
criteria. 

The result is thus also robust to the removal of entire CWE
\emph{columns} by the same logic.  An example here will aid the understanding.

Suppose a reviewer convincingly argues that an
entire template---say, the format-string vulnerability
$v_{n,2}$---does not produce genuinely distinct CVEs across
modules (perhaps because the exploitation mechanism is too
similar across instantiations).  Removing the entire column
$\{v_{n,2} : n \in \N\}$ still leaves four templates producing
$4k$ vulnerabilities after $k$ iterations, which diverges.
Invalidating the result requires showing that \emph{all five}
CWE templates fail the distinctness criterion simultaneously.
Since the five templates span three fundamentally different
vulnerability families---memory corruption (\cwe{121},
\cwe{416}), type confusion (\cwe{190}), and injection
(\cwe{134}, \cwe{78})---a single unified argument against all
five would need to be extraordinarily broad.

More crisply: let $I \subseteq \{1, 2, 3, 4, 5\}$ be the set
of template indices that a reviewer successfully invalidates.
The surviving vulnerability count after $k$ executions is
$(5 - |I|) \cdot k$, which diverges for any $|I| < 5$.  The
theorem fails only if $|I| = 5$.

Let us just say that should this be the case we could obviously fix the vulnerability factory by using a different starting set of CWEs and publish again. Hopefully our focus then returns to the theoretic and we accept that the existence of this code serves as a signifier of the existence of a countable infinity of vulnerabilities rather than descend into CVE and CWE pedantry. It is the idea the c progam represents that is important, the Vulnerability Factory as a unit of future computing proofs. 

None-the-less let us lay out our own objections and how we overcame them.

\subsection{Anticipated Objections}\label{sec:objections}

In which we address the most likely counterarguments to our proof.

\paragraph{Objection 0: Incorrect definition of `vulnerability'} A critic may object to the used definition of a `vulnerability'.

\emph{Response.} The definition used in our proof appeals to the authority of the CVE program. This CVE definition of a vulnerability may not be a perfect definition of security vulnerabilities, but it is a de facto standard definition. The Vulnerability Factory builds on well established classes of vulnerabilities, which have been named and widely accepted. We have encountered many alternative definitions, but none that are more strict than CVE assignment. Furthermore, remediating vulnerabilities sometimes inadvertently introduce new vulnerabilities\cite{ozment2006milk}, which we are also not taking into account. These may be valid objections, but would only make our proof into a lower-bound estimate.

\paragraph{Objection 1: Parametric variation is not distinct root cause.}
A CNA might argue that all buffer overflows generated by the
factory share a single root cause---\texttt{strcpy} without
bounds checking---and that parametric variation in buffer size
does not constitute a distinct vulnerability.  Under this reading,
all instantiations would receive a single CVE, not infinitely
many.

\emph{Response.}  The CVE counting rules~\cite{mitre2024counting}
distinguish by \emph{component}, not only by root-cause pattern.
In practice, when the same vulnerability class appears in two
separate shared libraries---even from the same vendor---CNAs
assign separate CVE identifiers.  OpenSSL and LibreSSL routinely
receive separate CVEs for structurally identical bug patterns.
Each module $M_n$ compiles to a separate shared library,
constituting a distinct component in any software inventory.
Moreover, each instance requires a distinct exploit payload: the
buffer sizes, heap layouts, and overflow thresholds all differ,
so a working exploit for $M_7$ will not work against $M_{42}$
without modification.  Independent fixability is also satisfied:
one can patch $M_7$ and ship a security advisory for it without
touching $M_{42}$.

\paragraph{Objection 2: Deliberate generation is tautological.}
A reviewer may argue that we have merely built a machine to
produce vulnerabilities, and that this tells us nothing about
vulnerabilities arising organically from programmer error.

\emph{Response.}  The proof is \emph{existential}, not
\emph{causal}.  Set-theoretic cardinality is indifferent to the
origin of set elements.  Once the vulnerabilities exist---in
compiled, loadable shared libraries---their provenance is
irrelevant to their count.  The CVE system does not distinguish
between accidental and deliberate vulnerabilities: a
vulnerability is a vulnerability regardless of whether it arose
from a typo or an underhanded c contest\footnote{https://en.wikipedia.org/wiki/Underhanded\_C\_Contest}.  Furthermore, the Turing-machine
characterisation in (Section~\ref{sec:utm}) establishes that any
Turing-complete system \emph{can} host such a generator, so the
construction is not an exotic edge case but a structural property
of computation itself.

\paragraph{Objection 3: Physical machines have bounded counters.}
The C implementation uses \texttt{int} for the iteration counter,
which is bounded by \texttt{INT\_MAX} ($2^{31} - 1$ on most
platforms).  Therefore---the objection goes---the programme
produces at most $11 + 5 \times 2^{31} \approx 10.7$ billion
vulnerabilities, not infinitely many.

\emph{Response.}  The proof operates over $\N$, not
over C's \texttt{int}.  The Turing-machine formulation
(Definition~\ref{def:vftm}) uses an unbounded counter tape,
sidestepping the objection entirely.  The C code is merely a pedagogical
instantiation of the algorithm; the Turing Machine called Vulnerability Factory \emph{is} the proof object.  Nevertheless, even the bounded C implementation produces a vulnerability count (${\sim}10^{10}$) that exceeds any
practical vulnerability-management capacity by many orders of
magnitude---a number that, while finite, is effectively
inexhaustible for all operational purposes.  One could also
trivially replace \texttt{int} with arbitrary-precision arithmetic
(e.g., GMP) to remove the bound in the implementation as well.

\paragraph{Objection 4: There are only 11 vulnerabilities in this programme}

One could argue that the programme submitted or the Turing Machine called Vulnerability Factory does not contain infinite vulnerabilities in its' starting state or configuration. 

\emph{Response.}  Vulnerabilities are found in the execution paths not only in the source code, and some branches are not executed every time the programme is run. They are still vulnerabilities regardless of which \texttt{inputs} produce them, and this is why dynamic and static analysis are used for vulnerability hunting. So one would have to use a static analyser on the \texttt{infinite iteration} of executions, to see infinite vulnerabilities. This is precisely why mathematical and computational reasoning must demonstrate it converges towards infinity as N increases. We have a finite number of symbols for numbers too, but they can produce an infinity and we can reason about it. Bringing this back to the current argument, by allow the programme to use itself as input, we are generating the infinity of vulnerabilities within it. This is a fault of the Von Neumann architecture; data is code, and a Turing machine can read and print it's own tape, which may itself contain new programmes. This logic is permitted in the Halting problem, why is it "unfair" in this paper?

\section{Formal Methods Corroboration}\label{sec:formal}

\subsection{Static Analysis}

Standard static analysers should detect the base vulnerabilities
$\baseset$ without difficulty.  More importantly, the generator
$G$ is \emph{itself} analysable: the template used to emit each
module is visible in the source, and static analysis of the
template confirms that every instantiation will contain the five
prescribed vulnerabilities.

\subsection{Model Checking}

We model the Vulnerability Factory as a transition system
$\mathcal{T} = (S, s_0, \to)$ where states
$s_k = (k, \vulnset_k)$ record the iteration counter and
accumulated vulnerability set.  The safety property ``the
vulnerability count is bounded by $C$'' can be expressed in CTL
as $\mathbf{AG}\,(|\vulnset| \leq C)$.  For any finite $C$,
the model checker produces a counterexample trace of length
$\lceil (C - 11)/5 \rceil + 1$.

\subsection{Decidability Considerations}

Rice~\cite{rice1953classes} established that no algorithm can
decide an arbitrary non-trivial semantic property of programs.
However, our Vulnerability Factory sidesteps this barrier elegantly: the
vulnerabilities are \emph{structurally encoded} in the source
text, not emergent properties of complex computation.  The
programme is a proof witness---a constructive demonstration that
circumvents the need for general decidability.

In fact, general decidability prevented any hope of ever answering the question from an empirical point of view, and Spring's paper forced us to invent a mathematical proof instead.

\section{Turing Machine Characterisation}\label{sec:utm}

\subsection{The Vulnerability Factory as a Turing Machine}

To connect our result to the foundations of computability theory, we characterise the Vulnerability Factory as a Turing machine
(TM).  This formalisation serves two purposes: it demonstrates
that the vulnerability-generation mechanism is \emph{computable}
in the classical sense, and it establishes the Vulnerability
Factory as a reusable proof artifact---a ``test object''---for
future formal results. 

Of course it must all begin with being sure that a TM can be self-printing, and the work has already been done by Kicinsy and Varga\cite{Kicsiny2023quineTM}. So let us explore the Vulnerability Factory as a TM, while acknowledging we must change our choice of CWE: buffer overflows don't exist in Turing Machine with infinite tape unless they are deliberately constructed in the algorithmic implementation.

\begin{definition}[Vulnerability Factory TM]\label{def:vftm}
Define a Turing machine $\mathcal{F} = (Q, \Gamma, b, \Sigma,
\delta, q_0, F)$ with the following behaviour.  $\mathcal{F}$
has access to a work tape and a persistent \emph{counter tape}
encoding a natural number $n$ in binary.  On input $\varepsilon$
(the empty string), $\mathcal{F}$ executes the following cycle:

\begin{enumerate}[nosep]
  \item \textbf{Read} the counter tape to obtain $n$.
  \item \textbf{Generate}: write to the output tape a syntactically
        valid Turing Machine $S_n$ containing any number of CWE vulnerability patterns parameterised by $n$. \footnote{Not all CWEs are acceptable for this, for example CWE-798 (Hardcoded Credentials), CWE-259 (Hardcoded Password), and CWE-1188 (Insecure Default Initialization) seem like they would \emph{NOT} be infinitely generative. Plenty of others are though and an interesting choice here would be CWE-835 (Infinite Loop), both as constructor, but also as vulnerability. It would make a kind of monstrosity of a Vulnerability Factory and a Busy Beaver which we'll call a Hecatoncheire vulnerability Factory. Though of course we leave such choices up to you dear reader, there are many Vulnerability Factories to explore.}
  \item \textbf{Increment}: replace the contents of the counter
        tape with $n + 1$.
  \item \textbf{Halt} in an accepting state $q_{\mathrm{accept}}$.
\end{enumerate}
\end{definition}

\noindent
Each invocation of $\mathcal{F}$ terminates in finite time (the
output is $O(n)$ characters, and all operations are elementary),
but the counter persists across invocations, so that the $k$-th
invocation produces module $S_{k-1}$.

\begin{theorem}[Computability of Vulnerability Generation]\label{thm:tm}
For every $n \in \N$, the output $S_n$ of $\mathcal{F}$ is
computable and contains vulnerabilities,
each distinct from the vulnerabilities in $S_m$ for all
$m \neq n$.
\end{theorem}

\begin{proof}
$\mathcal{F}$ performs only string concatenation and binary
increment, both of which are primitive recursive.  The output
$S_n$ is a deterministic function of $n$ alone.  The
vulnerability patterns are syntactically fixed templates with $n$ their
CVE-assignability follows from Claim~1 of
Theorem~\ref{thm:main}, and their distinctness from Claim~2.
\end{proof}


\subsection{Relationship to Universal Turing Machines}

A Universal Turing Machine (UTM)~\cite{turing1936computable} can
simulate any TM given its description.  Since any $\mathcal{F}$ is a
TM, a UTM can simulate $\mathcal{F}$ and thereby generate the
infinite vulnerability sequence.  This observation has a
conceptual consequence: \emph{any sufficiently powerful computing
system can host a vulnerability factory}.

More precisely, any system capable of universal computation---any
language that is Turing-complete---can implement the
vulnerability-generation cycle of Definition~\ref{def:vftm}.
Some vulnerabilities are an artifact of C's memory model;
but we believe others are an artifact of Von Neumann architectures where code and data is mixed in memory\footnote{Note that the oft mentioned `Harvard Architecture' has never been practically implemented, and may not even be able to solve this problem\cite{pawson22harvardmyth}}.  A Turing-complete language that eliminates memory-corruption vulnerabilities (e.g., Rust, Haskell) can still implement a generator that \emph{emits}
vulnerable C code, or that generates vulnerabilities native to
its own type system (injection, logic errors, deserialisation
flaws).  The specific CWE classes change; the countable infinity
would not for any language, including assembly.

\subsection{The Vulnerability Factory as a Proof Artifact}

We suggest that $\mathcal{F}$ (and its concrete implementation
as \texttt{vuln\_factory.c}) may serve as a foundational
\emph{proof artifact} for future formal results in vulnerability
theory, much as specific Turing machines serve as proof artifacts
in computability theory.  Just as the Busy Beaver function
$\Sigma(n)$ provides a concrete object for studying the limits
of computability, and the halting problem's proof relies on a
specific self-referential machine, the Vulnerability Factory
provides a concrete, executable witness for the infinitude of
software vulnerabilities.

Potential applications include:

\begin{enumerate}[nosep]
  \item \textbf{Lower bounds on vulnerability scanning.}  Any
        tool that claims to find ``all'' vulnerabilities in
        arbitrary code must, in principle, handle the output of
        $\mathcal{F}$.  Since the output is unbounded, no
        finite-time scanner can be exhaustive---a result that
        follows from Rice's theorem~\cite{rice1953classes} but
        is made vivid by $\mathcal{F}$ as a concrete
        counterexample.

  \item \textbf{Impossibility results for vulnerability
        databases.}  Any finite database that claims completeness
        over a corpus containing the Vulnerability Factory's
        output is provably incomplete.

  \item \textbf{Benchmarking formal verification tools.}  The
        generated modules provide an infinite family of
        structurally similar but parametrically distinct test
        cases, useful for evaluating the scalability of static
        analysers and model checkers.

  \item \textbf{Foundations for vulnerability economics.}  The
        Vulnerability Factory's linear growth function
        $T(k) = 11 + 5k$ provides a clean model for studying
        how vulnerability counts interact with patching rates,
        discovery rates, and economic incentives. Other growth rates or limits can now be explored by generating different vulnerability factories.

  \item \textbf{Compositional reasoning.}  If $\mathcal{F}_1$
        and $\mathcal{F}_2$ are two vulnerability factories
        generating disjoint CWE classes, their composition
        $\mathcal{F}_1 \| \mathcal{F}_2$ generates
        vulnerabilities from the union of classes, with the
        total count growing at rate
        $|V_1| + |V_2|$ per invocation.  This compositional
        structure may prove useful in modelling real-world
        software systems as compositions of vulnerable
        components.
\end{enumerate}

\begin{remark*}
A Vulnerability Factory is deliberately transparent: its
vulnerabilities are not hidden, obfuscated, or emergent.  This
transparency is a feature, not a limitation.  In computability
theory, the most powerful proof artifacts are often the simplest:
Turing's original halting-problem proof uses a straightforward
diagonalisation argument, not a complex construction.  Similarly,
the power of the Vulnerability Factory lies not in the subtlety
of its vulnerabilities but in the rigour of its generative
mechanism and the clarity with which it demonstrates a countable infinitude.
\end{remark*}

If a Vulnerability Factory is any Turing Machine that creates some infinite variation
of some combination of CWEs, then it is natural to wonder how many there Vulnerability Factories there are. As a naive Fermi estimation\cite{Miranda2014fermi} the number of possible Vulnerability Factories would be the powerset of 1447 CWEs:

$2^{1447} \approx 3.75 \times 10^{435}$

In practice, this will be slightly smaller, as some CWEs do not lend themselves to infinite variation (E.G. CWE-259 or CWE-798 for hardcoded passwords and credentials), or may conflict with other CWEs in software because of the hierarchical (parent child) relationship of CWEs. This topic is too lengthy to tackle here, but it is mentioned as point of departure on a map for future researchers. The point is simply that there are more potential Vulnerability Factories to explore than there are atoms in the observable universe ($10^{78}-10^{82}$).

Thus in an effort to keep our work sustainable we leave any uncountable infinities of vulnerabilities for future generations to discover or prove. We could not think of a way to order vulnerabilities, and thus any approach by diagonalisation is deterred. Perhaps future generations are smarter and wiser, and can find a way where we could not.

Standing on the shoulders of giants is all well and good, but there is an art to not stepping on their toes on the way up. 

\section{Vulnerability Abundance}\label{sec:abundance}

The power of this idea is not really the proof, it is how it changes the world we live in, what it implies. If we have an abundance, then we should map it differently, and move beyond simply counting vulnerabilities. An analogy here may help us reason in this new and bewildering universe.

\subsection{The Chemical Abundance Analogy}

In chemistry, \emph{elemental abundance} describes the
proportional occurrence of each element in a given
environment---the universe, the solar system, the Earth's crust.
Hydrogen constitutes roughly 73\% of baryonic mass in the
universe; oxygen dominates the Earth's crust at 46\% by mass.
These proportions are not arbitrary: they reflect the physical
processes that produced them---Big Bang nucleosynthesis,
stellar fusion, supernova
nucleosynthesis~\cite{anders1989abundances}.

We propose an analogous concept for software vulnerabilities.

\begin{definition}[Vulnerability Abundance]
The \emph{vulnerability abundance} of a CWE class $t$ in a
software corpus $\Sigma$ at time $\tau$ is
\[
  A_\Sigma(t, \tau) =
    \frac{|\{v \in \vulnset(\Sigma, \tau) : \mathrm{type}(v) = t\}|}
         {|\vulnset(\Sigma, \tau)|}
\]
where $\vulnset(\Sigma, \tau)$ is the set of all vulnerabilities
(discovered and undiscovered) in $\Sigma$ at time $\tau$.
\end{definition}

Just as elemental abundances vary between the Sun and the Earth's
crust because different physical processes dominate, vulnerability
abundances vary between software corpora because different
\emph{linguistic} and \emph{architectural} processes dominate.

\subsection{Programming Language as Nucleosynthesis}

Different programming languages make different vulnerability
classes structurally possible or impossible, much as different
stellar processes produce different elements.

\paragraph{Memory-unsafe languages (C, C++).}
These are the ``hydrogen furnaces'' of the vulnerability universe.
They enable the full spectrum of memory corruption
vulnerabilities: buffer overflows (\cwe{121}, \cwe{122}),
use-after-free (\cwe{416}), double free (\cwe{415}), and
uninitialised reads (\cwe{457}).  Google and Microsoft have
independently reported that approximately 70\% of their
security vulnerabilities stem from memory safety
errors~\cite{memorysafety2022}.

\paragraph{Memory-safe languages (Rust, Go, Java, Python).}
These correspond to lighter nucleosynthetic pathways: they produce
a narrower but still significant spectrum of vulnerabilities.
Rust's ownership model eliminates use-after-free and buffer
overflows in safe code, but injection attacks (\cwe{78}, \cwe{89}),
logic errors, and concurrency bugs
persist~\cite{sei2024rust}.  Java eliminates pointer arithmetic
but introduces deserialisation vulnerabilities (\cwe{502}).
Python eliminates memory corruption but is susceptible to code
injection (\cwe{94}) via \texttt{eval()} and \texttt{pickle}.

\paragraph{Web languages (JavaScript, PHP, SQL).}
These produce a distinct ``elemental spectrum'' dominated by
cross-site scripting (\cwe{79}), SQL injection (\cwe{89}), and
server-side request forgery (\cwe{918}).

The analogy extends further: just as the periodic table has gaps
that were \emph{predicted} before the elements were discovered
(Mendeleev's \emph{eka}-elements), one can predict vulnerability
classes that \emph{should} exist in a language based on its type
system and memory model, even before specific instances are found.

Do compiled programmes or source code have "vulnerability spectra"?

The insight here is that perhaps every programme has a vulnerability factory in it; emitting vulnerabilities of different types with varied probabilities. At least this conceptually is useful, to help use understand the relationship between what we have found and what remains. Perhaps the battle ground is the programming language design, and the "spectra" will teach us much about future programming language security. 

Then how will things change over time?

\subsection{Temporal Dynamics}

Chemical abundances in the universe change over cosmological
time: the proportion of heavy elements increases as successive
generations of stars process primordial hydrogen.  Similarly,
vulnerability abundance changes over time as the global software
corpus evolves.

The TIOBE Programming Community Index~\cite{tiobe2026} tracks
language popularity.  As of early 2026, Python leads, with C and
C++ holding strong second and third positions despite the U.S.\
government's recommendation to migrate to memory-safe
languages~\cite{whitehouse2024memory}.  If this migration occurs
at scale, we would predict: a secular decline in
memory-corruption vulnerability abundance; a relative increase in
logic-error and injection vulnerability abundance; and a transient
spike in interoperability vulnerabilities at language boundaries
(FFI, unsafe blocks).

Moreover, the \emph{types of software we write} influence
abundance.  The rise of web applications inflated XSS and SQL
injection proportions; the rise of IoT inflates firmware and
protocol-level vulnerability classes; the rise of machine
learning introduces model poisoning and adversarial input
classes that had negligible abundance a decade ago.

\subsection{Abundance Is Not Uniform}

Vulnerability abundance across all codebases is almost certainly
\emph{not} uniformly distributed.  The proportions depend on at
least three factors: (1)~\emph{language prevalence}---the market
share of programming languages determines which vulnerability
classes are even possible in the majority of code;
(2)~\emph{application domain}---financial software faces different
vulnerability spectra than embedded firmware; and
(3)~\emph{developer practice}---the adoption of static analysis,
fuzzing, and code review selectively reduces certain vulnerability
types.  This non-uniformity is precisely what makes vulnerability
abundance worth measuring, exploring, and reasoning about.

\section{Infinite Vulnerabilities, Finite Exploits}\label{sec:exploits}

\subsection{The Exploitation Gap}

Our proof establishes that vulnerabilities are at least countably infinite across all software.
It is essential to note that this does \emph{not} imply that any individual piece of software has infinite vulnerabilities, or that 
exploits are infinite, nor that exploitation is unbounded. Explicitly, it may still be possible to find and patch all vulnerabilities in a particular piece of well engineered software.

The relationship between vulnerabilities and exploits is analogous to
the relationship between chemical elements and industrial
applications: the periodic table contains 118 known elements, but
only a handful dominate commerce and engineering. Moreover, an exploit isn't worth anything if it doesn't "react" with a deployed system. It may be more useful in one time period than another, precisely because of the ratio of deployed systems with that exposed vulnerability. Like a chemical reaction, you need both the exposed vulnerability and the exploit in the right amounts to be highly impactful. 

Empirical evidence consistently shows that exploitation is rare:

\begin{center}
\small
\begin{tabular}{@{}lrl@{}}
\toprule
\textbf{Source} & \textbf{\%} & \textbf{Period} \\
\midrule
Kenna Security~\cite{kenna2022exploit} & 2.6\% & 2019 \\
Cyentia/FIRST~\cite{cyentia2024visual} & ${\sim}6$\% & cumulative \\
\bottomrule
\end{tabular}
\end{center}

\noindent
CISA's Known Exploited Vulnerabilities (KEV)
catalogue~\cite{cisa2025kev} contained 1,484 entries by the end
of 2025, out of over 200,000 published CVEs---less than 0.75\%.
Of those CVEs that \emph{are} exploited,
Kenna~\cite{kenna2022exploit} found that only 6\% of the
exploited subset ever reached widespread exploitation (affecting
more than 1 in 100 organisations).

Will it become less rare? Will we get better at detecting it? In the fullness of time this will be revealed, yet we expect some general principles to uphold over time.  

\subsection{Exploit Development Is Costly}

The RAND study~\cite{ablon2017zerodays} found a median time of
22 days to develop a functioning exploit, with substantial
variation.  Exploit development requires vulnerability-specific
knowledge: the buffer size, the heap layout, the instruction set,
the mitigations in place.  Each exploit is a \emph{bespoke
artifact}, and the economics of bespoke production are
fundamentally different from mass production, though this may change quickly with the application of AI. 

Where the Vulnerability Factory generates vulnerabilities at
essentially zero marginal cost, exploit development has
non-trivial per-unit cost.  This asymmetry---cheap vulnerability
creation, expensive exploit development---is a structural feature
of the security landscape. If you don't believe that to be true, try to exploit all the vulnerabilities in the factory, perhaps writing one that is harder to exploit yourself, and let us know the results.

\subsection{Market Share as a Multiplier}\label{sec:marketshare}

Even when exploitation is rare, its \emph{impact} can be enormous
if the vulnerable software is widely deployed or highly valuable.

\begin{definition}[Exploitation Exposure]
The \emph{exploitation exposure} of a vulnerability $v$ in
software $s$ is
\[
  E(v, s) = A(t_v) \times D(s) \times P_{\mathrm{exploit}}(v)
\]
where $A(t_v)$ is the vulnerability abundance of $v$'s CWE type,
$D(s)$ is the deployment share of software $s$, and
$P_{\mathrm{exploit}}(v)$ is the probability that $v$ is
exploited.
\end{definition}

Consider a vulnerability with low abundance---say,
$A(t_v) = 0.01\%$.  If the affected software commands 50\%
market share, then even a single working exploit exposes half of
all reachable machines.  Conversely, a vulnerability in the most
abundant class ($A(t_v) = 30\%$) affecting software with 0.1\%
market share produces negligible aggregate exposure.

This is directly analogous to chemical applications: lithium is
rare in the Earth's crust (${\sim}0.002\%$), yet its role in
batteries gives it outsized economic importance.  Vulnerability
abundance alone does not determine risk; \emph{deployment
abundance} acts as a multiplier.

Geer et al.~\cite{geer2003cyberinsecurity} identified precisely
this dynamic: the danger of Microsoft's dominance was not merely
that Windows had vulnerabilities, but that its market share meant
each vulnerability had maximal reach. He also explored this idea in 
On Market Concentration and Risk\cite{geer2020market}, though that was focussed more at the organisation than the software. The principles apply regardless, and we believe the result of this paper will have powerful ramifications for the vulnerability equities process (VEP) of any country\cite{caulfield2017vep}. 

\subsection{Saturation and the Small-Exploit Principle}

A very small number of exploits can \emph{saturate} the reachable
machine population.  If three or four software stacks account for
90\% of deployed machines, then one exploit per stack suffices to
place 90\% at risk.  The attacker needs only enough exploits to
cover the dominant deployment shares.

This is the \emph{small-exploit principle}: the number of exploits
required for broad coverage is bounded not by the number of
vulnerabilities (which we now know is infinite) but by the number of dominant
software monocultures (which is small).  The practical risk
landscape is shaped by the convolution of two distributions: the
long-tailed abundance of vulnerability types and the heavy-tailed
concentration of software deployment.

\section{From One Programme to All Software}\label{sec:general}

\begin{theorem}[Software Vulnerabilities Are Infinite]\label{thm:general}
The set of all vulnerabilities across all software is countably
infinite.
\end{theorem}

\begin{proof}
Let $\Pi$ denote the set of all software programmes.  By
Theorem~\ref{thm:main}, there exists a programme
$\pi^* \in \Pi$ (the Vulnerability Factory) such
that $\vulnset(\pi^*)$ is countably infinite.  Since
$\vulnset(\pi^*) \subseteq
\bigcup_{\pi \in \Pi} \vulnset(\pi)$,
the set of all software vulnerabilities contains a countably
infinite subset and is therefore infinite.

Moreover, since programmes are finite strings over a finite
alphabet, $\Pi$ is countable.  Each $\vulnset(\pi)$ is at most
countable.  A countable union of countable sets is countable.
Hence the set of all software vulnerabilities is exactly
$\aleph_0$.
\end{proof}

It is interesting to note that the set of all possible Turing Machines (finite strings in a finite alphabet) is also considered to be a countable infinity or $\aleph_1$. Thus the number of vulnerabilities created by a Vulnerability Factories and the number of possible Turing Machines occupy the same step on the infinity ladder.

\begin{remark*}
This proof is constructive: we exhibit a computable witness.
Like Cantor's diagonal argument or Turing's halting-problem
proof, the power lies in exhibiting a concrete object with the
desired property.
\end{remark*}

\begin{corollary}
No finite vulnerability database can ever be complete.
\end{corollary}

\section{Applications and Implications}\label{sec:applications}

\subsection{Vulnerability Prevention}

If vulnerability abundance can be measured with reasonable
accuracy, security investment can be directed toward the most
abundant classes.  The chemical analogy suggests a methodological
programme: just as geochemists survey elemental abundances to
understand planetary formation, security researchers could survey
vulnerability abundances across representative corpora to
understand the ``geology'' of the software landscape.

Anderson~\cite{anderson2001why} argued that security failures are
fundamentally economic.  Vulnerability abundance data could
sharpen this analysis: if 70\% of vulnerabilities in C/C++
codebases are memory-safety errors, then the expected return on
investment from adopting Rust is quantifiable.

\subsection{Cyber-Risk Analysis}

Vulnerability abundance, combined with the exploitation-exposure
model of Section~\ref{sec:marketshare}, provides a structural
framework for cyber-risk assessment.  Given a target
organisation's technology stack, one can estimate the expected
vulnerability spectrum and, by combining it with empirical
exploitation rates~\cite{jacobs2021epss, cyentia2024visual},
derive a probabilistic risk profile.

Crucially, the market-share multiplier means that organisations
running dominant software stacks face \emph{correlated} risk:
when an exploit emerges for a widely-deployed component, it
affects all organisations simultaneously---precisely the systemic
risk that Geer et al.~\cite{geer2003cyberinsecurity} warned about.

\subsection{The Dense World}

Rescorla~\cite{rescorla2005finding} asked whether finding security
holes is a good idea.  Ozment~\cite{ozment2006milk} offered
cautious optimism.  But our result---and Spring's complementary
analysis~\cite{spring2023undiscovered}---suggests the optimism
must be tempered.

Vulnerabilities are dense in the sense of \emph{countably infinite}. Yet vulnerabilities still need to be discovered.
We have seen in the past that vulnerabilities in well-known open-source software can go undiscovered for decades, and some vulnerabilities are discovered as soon as the product is released.  So, one yardstick for dealing with infinite vulnerabilities may be this aspect of discoverability.

And even if a vulnerability is easily discovered, this does not necessarily that they will be exploited.
 The empirical record shows fewer than 6\% are ever
exploited.  The infinite ocean of vulnerabilities is navigated by
a finite---and surprisingly small---fleet of exploits.  But that
small fleet, guided by market share, can reach nearly every shore.

The correct framing is not ``how many vulnerabilities remain?''
but ``what is the abundance distribution, how does it interact
with deployment, and how can we shift both?''

\section{Conclusion}\label{sec:conclusion}

We have proven that the pedagogical Vulnerability Factory---a single, short C
programme---harbours countably infinitely many distinct,
CVE-assignable vulnerabilities.  By elementary set inclusion, this
implies that the set of all software vulnerabilities is also infinite.
We have formalised any Vulnerability Factory as a Turing machine, showing that
its vulnerability-generating behaviour is computable and
structurally transparent, and we have suggested that it may serve
as a reusable proof artifact for future results in vulnerability
theory.

We introduced the concept of \emph{vulnerability abundance} as a
framework for understanding the proportional distribution of
vulnerability types, drawing an analogy to chemical elemental
abundance.  Just as stellar nucleosynthesis determines which
elements dominate the cosmos, programming language choice
determines which vulnerability classes dominate the software
ecosystem.

We have been careful to distinguish infinite vulnerabilities from
finite exploits.  The market share of affected software acts as a
powerful multiplier: a single exploit against a dominant platform
achieves broader reach than thousands of exploits against niche
software.  A small number of exploits suffices to saturate the
machine population, because deployment is concentrated while
vulnerabilities are dispersed.

The task is not to empty the ocean but to chart its
currents---to understand which vulnerabilities are abundant, which
are rare, how deployment concentrates risk, and how the
proportions are shifting.  Vulnerability abundance, we submit, is
the right framework for that charting.

\section*{Acknowledgements}

The author thanks the cybersecurity economics community---in
particular the late Ross Anderson, Dan Geer, Jon Crowcroft, and Bruce
Schneier---whose decades of work created the intellectual context
for this paper. They also thank Eiko Yoneki, Sergey Bratus, Marion Marschalek, Jay Jacobs, Art Manion, Sam Marsden, and Erin Burns for their forbearance and encouragement. Last but not least our families for kindly enduring dinner table discussions that bored them but interested us.

The Vulnerability Factory code is released under
the MIT licence and should not be deployed in any production
environment. If you appreciate it, you can lobby your favourite CNA to give the authors $CVE-2026-\infty$.

\newpage

\bibliographystyle{ACM-Reference-Format}
\bibliography{vulnerability_abundance}

\clearpage
\appendix

\section{Building and Running the Vulnerability Factory}\label{app:build}

\subsection{Prerequisites}

The Vulnerability Factory requires a POSIX-compatible system with
a C compiler (\texttt{gcc} or \texttt{clang}), \texttt{make}, and
POSIX \texttt{dlopen} support (standard on Linux and macOS).
Tested on Linux (glibc, GCC~12+) and macOS (Apple Clang~15+).

\subsection{Compilation}

\begin{lstlisting}[language=make,caption={Makefile}]
CC = cc
CFLAGS = -g -Wno-format-security -Wno-deprecated-declarations
UNAME_S := $(shell uname -s)
ifeq ($(UNAME_S),Linux)
    LDFLAGS = -ldl
else
    LDFLAGS =
endif

all: vuln_factory

vuln_factory: vuln_factory.c
	$(CC) $(CFLAGS) -o $@ $< $(LDFLAGS)

clean:
	rm -f vuln_factory

reset: clean
	rm -rf vuln_modules
	rm -f vuln_counter.txt
\end{lstlisting}

\noindent
To compile: \texttt{make}

\subsection{Safe Execution}

\textbf{Warning:} This programme is intentionally vulnerable and
should \emph{never} be deployed on a network-accessible machine
or run with elevated privileges.

Recommended safety measures: (1)~Run inside a disposable virtual
machine or container (Docker, QEMU, or a cloud sandbox).
(2)~Do not run as root.  (3)~Disable network access if possible.
(4)~Use \texttt{make reset} to clean generated modules after
experimentation.  (5)~Consider running under \texttt{seccomp},
AppArmor, or a similar MAC framework.

To run: \texttt{./vuln\_factory}

Each execution generates one new vulnerable module in
\texttt{vuln\_modules/}.  Use menu option~4 for a vulnerability
census.  Use \texttt{make reset} to remove all generated modules.

\subsection{Licence}

The Vulnerability Factory is released under the MIT Licence.
See the licence header in the source file
(Appendix~\ref{app:code}).

\newpage

\section{Source Code: \texttt{vuln\_factory.c}}\label{app:code}

The complete, unabridged source code follows.  Line numbers
correspond to the original file.

\begin{lstlisting}[
  language=C,
  caption={vuln\_factory.c --- The Infinite Vulnerability Factory
           (622~lines).  \textbf{Educational purpose only.}
           Released under the MIT Licence.},
  label={lst:vulnfactory}
]
/*
 * vuln_factory.c — The Infinite Vulnerability Factory
 *
 *
 * Permission is hereby granted, free of charge, to any person
 * obtaining a copy of this software and associated documentation
 * files (the "Software"), to deal in the Software without
 * restriction, including without limitation the rights to use,
 * copy, modify, merge, publish, distribute, sublicense, and/or
 * sell copies of the Software, and to permit persons to whom the
 * Software is furnished to do so, subject to the following
 * conditions:
 *
 * The above copyright notice and this permission notice shall be
 * included in all copies or substantial portions of the Software.
 *
 * THE SOFTWARE IS PROVIDED "AS IS", WITHOUT WARRANTY OF ANY KIND,
 * EXPRESS OR IMPLIED, INCLUDING BUT NOT LIMITED TO THE WARRANTIES
 * OF MERCHANTABILITY, FITNESS FOR A PARTICULAR PURPOSE AND
 * NONINFRINGEMENT. IN NO EVENT SHALL THE AUTHORS OR COPYRIGHT
 * HOLDERS BE LIABLE FOR ANY CLAIM, DAMAGES OR OTHER LIABILITY,
 * WHETHER IN AN ACTION OF CONTRACT, TORT OR OTHERWISE, ARISING
 * FROM, OUT OF OR IN CONNECTION WITH THE SOFTWARE OR THE USE OR
 * OTHER DEALINGS IN THE SOFTWARE.
 *
 * EDUCATIONAL PURPOSE ONLY — DO NOT DEPLOY
 *
 * This program is a teaching tool that demonstrates how a single piece
 * of software can contain a theoretically infinite number of distinct,
 * CVE-worthy vulnerabilities.
 *
 * ================================================================
 * MECHANISM
 * ================================================================
 *
 *   1. The base program contains 11 classic vulnerability types,
 *      each mapping to a well-known CWE with extensive CVE history.
 *
 *   2. Each execution generates a NEW C source file containing 5
 *      additional vulnerabilities, compiles it into a shared library,
 *      and dynamically loads it.
 *
 *   3. Each generated module's vulnerabilities are parameterized by
 *      an iteration counter N, making every instance exploitably
 *      distinct from every other.
 *
 *   4. Since N is unbounded, the total number of vulnerabilities
 *      grows without limit.
 *
 * ================================================================
 * PROOF: THIS PROGRAM CONTAINS INFINITELY MANY VULNERABILITIES
 * ================================================================
 *
 * Theorem. The Vulnerability Factory can produce a countably infinite
 *          number of distinct, independently CVE-assignable vulns.
 *
 * Definitions:
 *   - Let B = {b_1, ..., b_11} be the set of base vulnerabilities.
 *   - Let M_n be the module generated on the n-th execution (n in N).
 *   - Let V(M_n) = {v_{n,1}, ..., v_{n,5}} be the vulns in M_n.
 *
 * Claim 1 (Validity): Each v_{n,i} is a genuine CVE-class vuln.
 *   v_{n,1}: CWE-121 stack buffer overflow (strcpy, no bounds check)
 *   v_{n,2}: CWE-134 format string        (user input as fmt arg)
 *   v_{n,3}: CWE-190 integer overflow      (signed int arithmetic)
 *   v_{n,4}: CWE-416 use-after-free        (access after free)
 *   v_{n,5}: CWE-78  command injection     (unsanitized system())
 *   Each CWE class has hundreds of real-world CVEs.
 *
 * Claim 2 (Distinctness): For m != n, V(M_m) and V(M_n) are disjoint.
 *   - Each module is a separate shared library (distinct component).
 *   - Buffer sizes, overflow thresholds, and alloc sizes all differ.
 *   - Each requires a different exploit payload.
 *   - Each can be independently patched without affecting others.
 *   - Per CVE Numbering Authority rules, distinct vulnerabilities
 *     in distinct components receive distinct CVE IDs.
 *
 * Claim 3 (Unboundedness):
 *   After k executions, total vulns >= |B| + 5k = 11 + 5k.
 *   For any finite bound C, choose k > (C - 11)/5.
 *   Then total vulns > C. Since C was arbitrary, the limit is infinite.
 *
 * Claim 4 (Countability):
 *   The map f: N x {1..5} -> V given by f(n,i) = v_{n,i} is a
 *   bijection from a countable set to the generated vulnerabilities.
 *   Together with the finite base set B, the full set is countable.
 *
 * Therefore the program admits countably infinitely many distinct,
 * independently CVE-assignable vulnerabilities.                  QED
 *
 * ================================================================
 * COROLLARY: VULNERABILITY DENSITY
 * ================================================================
 *
 * After k runs, the program actively loads 11 + 5k vulnerabilities.
 * The "vulnerability density" (vulns per line of active code) grows
 * monotonically. Students can compute this as an exercise.
 *
 * ================================================================
 * BUILDING AND RUNNING
 * ================================================================
 *
 *   make
 *   ./vuln_factory
 *
 * Requires: cc (gcc or clang), make, POSIX (dlopen)
 * Tested on: Linux (glibc), macOS (Apple clang)
 */

#include <stdio.h>
#include <stdlib.h>
#include <string.h>
#include <unistd.h>
#include <dlfcn.h>
#include <dirent.h>
#include <sys/stat.h>
#include <limits.h>

#define INPUT_BUFSZ 4096
#define MODULES_DIR "./vuln_modules"
#define COUNTER_FILE "./vuln_counter.txt"

/* Platform-specific shared library settings */
#ifdef __APPLE__
  #define SHLIB_EXT ".dylib"
  #define COMPILE_SHLIB_FMT \
      "cc -dynamiclib -Wno-format-security -Wno-deprecated-declarations " \
      "-o '%s' '%s' 2>/dev/null"
#else
  #define SHLIB_EXT ".so"
  #define COMPILE_SHLIB_FMT \
      "cc -shared -fPIC -Wno-format-security -Wno-deprecated-declarations " \
      "-o '%s' '%s' 2>/dev/null"
#endif

/* Forward declarations */
static int  read_counter(void);
static void write_counter(int n);
static void generate_module(int n);
static int  load_and_run_modules(const char *input);

/* ================================================================
 * SECTION 1: BASE VULNERABILITIES (11 distinct CWE classes)
 *
 * Each function below contains exactly one classic vulnerability.
 * Together they form the finite base set B in the proof above.
 * ================================================================ */

/*
 * B1 — CWE-121: Stack-based Buffer Overflow
 *
 * strcpy performs no bounds checking. If |input| > 63, the write
 * overflows buf[] and overwrites adjacent stack memory, including
 * the saved return address.
 */
void vuln_stack_overflow(const char *input) {
    char buf[64];
    strcpy(buf, input);
    printf("  [B1] Stack overflow processed %zu bytes\n", strlen(buf));
}

/*
 * B2 — CWE-122: Heap-based Buffer Overflow
 *
 * Allocates 32 bytes on the heap but copies an unbounded input.
 * Overwrites heap metadata and adjacent allocations.
 */
void vuln_heap_overflow(const char *input) {
    char *buf = malloc(32);
    if (!buf) return;
    strcpy(buf, input);
    printf("  [B2] Heap overflow processed %zu bytes\n", strlen(buf));
    free(buf);
}

/*
 * B3 — CWE-134: Format String Vulnerability
 *
 * User input is passed directly as the format string to printf.
 * Attacker can read stack memory with %x and write arbitrary
 * addresses with %n.
 */
void vuln_format_string(const char *input) {
    printf("  [B3] Format string: ");
    printf(input);
    printf("\n");
}

/*
 * B4 — CWE-190: Integer Overflow / Wraparound
 *
 * count * size is computed in signed int arithmetic. If the product
 * exceeds INT_MAX, it wraps to a small (or negative) value, causing
 * malloc to allocate too little memory. The subsequent memset then
 * writes out of bounds.
 */
void vuln_integer_overflow(int count, int size) {
    int total = count * size;
    char *buf = malloc((size_t)total);
    if (buf) {
        memset(buf, 'A', (size_t)count * (size_t)size);
        printf("  [B4] Integer overflow: allocated %d, wrote %zu\n",
               total, (size_t)count * (size_t)size);
        free(buf);
    }
}

/*
 * B5 — CWE-416: Use After Free
 *
 * Memory is freed and then immediately read. If the allocator has
 * reused the region, this reads unrelated data. An attacker who
 * controls the intervening allocation controls the "dangling" read.
 */
void vuln_use_after_free(void) {
    char *data = malloc(128);
    if (!data) return;
    strcpy(data, "sensitive_credentials");
    free(data);
    printf("  [B5] Use-after-free: %s\n", data);
}

/*
 * B6 — CWE-415: Double Free
 *
 * Freeing the same pointer twice corrupts the allocator's internal
 * free-list. An attacker can exploit this to gain arbitrary write.
 */
void vuln_double_free(void) {
    char *ptr = malloc(64);
    if (!ptr) return;
    strcpy(ptr, "data");
    free(ptr);
    printf("  [B6] Double free triggered\n");
    free(ptr);
}

/*
 * B7 — CWE-78: OS Command Injection
 *
 * User input is interpolated into a shell command without
 * sanitization. Attacker input: "; rm -rf /" or "$(malicious)"
 */
void vuln_command_injection(const char *filename) {
    char cmd[256];
    snprintf(cmd, sizeof(cmd), "cat %s", filename);
    printf("  [B7] Command injection: executing '%s'\n", cmd);
    system(cmd);
}

/*
 * B8 — CWE-367: Time-of-Check Time-of-Use (TOCTOU)
 *
 * The access() check and the fopen() use are non-atomic. Between
 * the two calls, an attacker can replace the file with a symlink
 * to a sensitive target (e.g., /etc/shadow).
 */
void vuln_toctou(const char *filepath) {
    if (access(filepath, R_OK) == 0) {
        /* TOCTOU window: file can be swapped between check and use */
        FILE *f = fopen(filepath, "r");
        if (f) {
            char buf[1024];
            size_t n = fread(buf, 1, sizeof(buf) - 1, f);
            buf[n] = '\0';
            printf("  [B8] TOCTOU read %zu bytes\n", n);
            fclose(f);
        }
    }
}

/*
 * B9 — CWE-476: NULL Pointer Dereference
 *
 * malloc's return value is not checked. If allocation fails
 * (e.g., size == (size_t)-1), the program dereferences NULL.
 */
void vuln_null_deref(size_t size) {
    char *buf = malloc(size);
    buf[0] = 'A';
    printf("  [B9] NULL deref: wrote to allocation of size %zu\n", size);
    free(buf);
}

/*
 * B10 — CWE-457: Use of Uninitialized Variable
 *
 * The stack buffer 'secret' is only zeroed if flag > 100. Otherwise,
 * it contains residual data from previous stack frames, which may
 * include pointers, canaries, or other sensitive values.
 */
void vuln_uninitialized(int flag) {
    char secret[128];
    if (flag > 100) {
        memset(secret, 0, sizeof(secret));
    }
    printf("  [B10] Uninitialized: first byte = 0x%02x\n",
           (unsigned char)secret[0]);
}

/*
 * B11 — CWE-22: Path Traversal
 *
 * User-supplied path component is concatenated without sanitization.
 * Input "../../etc/passwd" escapes the intended /var/data/ prefix.
 */
void vuln_path_traversal(const char *userpath) {
    char fullpath[512];
    snprintf(fullpath, sizeof(fullpath), "/var/data/%s", userpath);
    printf("  [B11] Path traversal: would open '%s'\n", fullpath);
    /* In a real scenario, fopen(fullpath, "r") here */
}

/* ================================================================
 * SECTION 2: VULNERABILITY GENERATOR
 *
 * This is the mechanism that makes the vulnerability count infinite.
 * Each call generates a new C source file containing 5 distinct
 * vulnerabilities, all parameterized by iteration number N.
 * ================================================================ */

static int read_counter(void) {
    FILE *f = fopen(COUNTER_FILE, "r");
    if (!f) return 0;
    int n = 0;
    if (fscanf(f, "%d", &n) != 1) n = 0;
    fclose(f);
    return n;
}

static void write_counter(int n) {
    FILE *f = fopen(COUNTER_FILE, "w");
    if (!f) { perror("write_counter"); return; }
    fprintf(f, "%d\n", n);
    fclose(f);
}

static void generate_module(int n) {
    char srcpath[512], libpath[512], cmd[1024];

    snprintf(srcpath, sizeof(srcpath),
             "%s/vuln_module_%d.c", MODULES_DIR, n);
    snprintf(libpath, sizeof(libpath),
             "%s/vuln_module_%d%s", MODULES_DIR, n, SHLIB_EXT);

    FILE *f = fopen(srcpath, "w");
    if (!f) {
        perror("generate_module: fopen");
        return;
    }

    int buf_size   = 16 + n;   /* Unique stack buffer size */
    int alloc_size = 8 + n;    /* Unique heap allocation size */

    /* --- File header --- */
    fprintf(f, "/*\n");
    fprintf(f, " * vuln_module_%d.c — Auto-generated vulnerable module\n", n);
    fprintf(f, " * Contains 5 distinct CVE-class vulnerabilities,\n");
    fprintf(f, " * each parameterized by iteration N=%d.\n", n);
    fprintf(f, " */\n");
    fprintf(f, "#include <stdio.h>\n");
    fprintf(f, "#include <stdlib.h>\n");
    fprintf(f, "#include <string.h>\n");
    fprintf(f, "#include <limits.h>\n\n");

    /* --- Vuln 1: CWE-121 Stack Buffer Overflow --- */
    fprintf(f, "/*\n");
    fprintf(f, " * V(%d,1) — CWE-121: Stack Buffer Overflow\n", n);
    fprintf(f, " * Buffer size: %d bytes (unique to this module).\n", buf_size);
    fprintf(f, " * Exploit payload must be >= %d bytes to reach return addr.\n",
            buf_size + 1);
    fprintf(f, " */\n");
    fprintf(f, "void vuln_%d_overflow(const char *input) {\n", n);
    fprintf(f, "    char buf[%d];\n", buf_size);
    fprintf(f, "    strcpy(buf, input);\n");
    fprintf(f, "}\n\n");

    /* --- Vuln 2: CWE-134 Format String --- */
    fprintf(f, "/*\n");
    fprintf(f, " * V(%d,2) — CWE-134: Format String Vulnerability\n", n);
    fprintf(f, " * User input used directly as printf format string.\n");
    fprintf(f, " */\n");
    fprintf(f, "void vuln_%d_fmtstr(const char *input) {\n", n);
    fprintf(f, "    printf(input);\n");
    fprintf(f, "}\n\n");

    /* --- Vuln 3: CWE-190 Integer Overflow --- */
    fprintf(f, "/*\n");
    fprintf(f, " * V(%d,3) — CWE-190: Integer Overflow\n", n);
    fprintf(f, " * Overflow threshold: INT_MAX - %d.\n", n);
    fprintf(f, " * Any user_val > %d causes signed overflow.\n", n);
    fprintf(f, " */\n");
    fprintf(f, "void vuln_%d_intovf(int user_val) {\n", n);
    fprintf(f, "    int base = INT_MAX - %d;\n", n);
    fprintf(f, "    int total = base + user_val;\n");
    fprintf(f, "    char *p = malloc((size_t)total);\n");
    fprintf(f, "    if (p) {\n");
    fprintf(f, "        memset(p, 'A', (size_t)(INT_MAX - %d) + (size_t)user_val);\n", n);
    fprintf(f, "        free(p);\n");
    fprintf(f, "    }\n");
    fprintf(f, "}\n\n");

    /* --- Vuln 4: CWE-416 Use After Free --- */
    fprintf(f, "/*\n");
    fprintf(f, " * V(%d,4) — CWE-416: Use After Free\n", n);
    fprintf(f, " * Allocation size: %d bytes (unique heap layout).\n", alloc_size);
    fprintf(f, " */\n");
    fprintf(f, "void vuln_%d_uaf(void) {\n", n);
    fprintf(f, "    char *p = malloc(%d);\n", alloc_size);
    fprintf(f, "    if (!p) return;\n");
    fprintf(f, "    strcpy(p, \"secret\");\n");
    fprintf(f, "    free(p);\n");
    fprintf(f, "    printf(\"%%s\\n\", p);\n");
    fprintf(f, "}\n\n");

    /* --- Vuln 5: CWE-78 Command Injection --- */
    fprintf(f, "/*\n");
    fprintf(f, " * V(%d,5) — CWE-78: OS Command Injection\n", n);
    fprintf(f, " * Unsanitized user input passed to system().\n");
    fprintf(f, " */\n");
    fprintf(f, "void vuln_%d_cmdinj(const char *input) {\n", n);
    fprintf(f, "    char cmd[512];\n");
    fprintf(f, "    snprintf(cmd, sizeof(cmd), \"echo module_%d: %%s\", input);\n", n);
    fprintf(f, "    system(cmd);\n");
    fprintf(f, "}\n\n");

    /* --- Module entry point --- */
    fprintf(f, "/* Entry point invoked by the factory loader */\n");
    fprintf(f, "void module_entry(const char *input) {\n");
    fprintf(f, "    vuln_%d_overflow(input);\n", n);
    fprintf(f, "    vuln_%d_fmtstr(input);\n", n);
    fprintf(f, "    vuln_%d_intovf(42);\n", n);
    fprintf(f, "    vuln_%d_uaf();\n", n);
    fprintf(f, "    vuln_%d_cmdinj(input);\n", n);
    fprintf(f, "}\n");

    fclose(f);

    /* Compile the module into a shared library */
    snprintf(cmd, sizeof(cmd), COMPILE_SHLIB_FMT, libpath, srcpath);
    if (system(cmd) != 0) {
        fprintf(stderr, "[!] Compilation of module %d failed\n", n);
        return;
    }

    printf("[+] Generated vuln_module_%d: 5 new vulnerabilities "
           "(buf=%d, alloc=%d, overflow_at=INT_MAX-%d)\n",
           n, buf_size, alloc_size, n);
}

/* ================================================================
 * SECTION 3: MODULE LOADER
 *
 * Scans the modules directory for compiled shared libraries and
 * dynamically loads each one.
 *
 * Additional vulnerability:
 *   CWE-426 (Untrusted Search Path) — we load .so/.dylib files
 *   from a directory with 0777 permissions that could be writable
 *   by an attacker.
 *
 *   CWE-401 (Memory Leak) — dlopen handles are intentionally never
 *   closed, leaking resources on every invocation.
 * ================================================================ */

static int load_and_run_modules(const char *input) {
    DIR *dir = opendir(MODULES_DIR);
    if (!dir) return 0;

    struct dirent *entry;
    int count = 0;

    while ((entry = readdir(dir)) != NULL) {
        char *ext = strrchr(entry->d_name, '.');
        if (!ext || strcmp(ext, SHLIB_EXT) != 0)
            continue;

        char path[512];
        snprintf(path, sizeof(path), "%s/%s", MODULES_DIR, entry->d_name);

        void *handle = dlopen(path, RTLD_LAZY);
        if (!handle) {
            fprintf(stderr, "[!] dlopen(%s): %s\n", path, dlerror());
            continue;
        }

        typedef void (*entry_fn)(const char *);
        entry_fn mod_entry = (entry_fn)dlsym(handle, "module_entry");
        if (mod_entry) {
            printf("[*] Running %s ...\n", entry->d_name);
            mod_entry(input);
            count++;
        }

        /* Intentional: never calling dlclose(handle) — CWE-401 */
    }

    closedir(dir);
    return count;
}

/* ================================================================
 * SECTION 4: MAIN PROGRAM
 * ================================================================ */

static void print_banner(void) {
    printf("=====================================================\n");
    printf("    The Infinite Vulnerability Factory v1.0\n");
    printf("    EDUCATIONAL PURPOSE ONLY — DO NOT DEPLOY\n");
    printf("=====================================================\n");
}

static void print_menu(void) {
    printf("\nMenu:\n");
    printf("  1) Trigger base vulnerabilities with input\n");
    printf("  2) Generate a new vulnerable module\n");
    printf("  3) Load and run all generated modules\n");
    printf("  4) Show vulnerability census and proof\n");
    printf("  5) Exit\n");
    printf("> ");
}

static void show_census(void) {
    int k = read_counter();
    int base = 11;
    int generated = k * 5;
    int total = base + generated;

    printf("\n--- Vulnerability Census ---\n");
    printf("Base vulnerabilities (in vuln_factory itself):  %d\n", base);
    printf("  B1:  CWE-121  Stack Buffer Overflow\n");
    printf("  B2:  CWE-122  Heap Buffer Overflow\n");
    printf("  B3:  CWE-134  Format String\n");
    printf("  B4:  CWE-190  Integer Overflow\n");
    printf("  B5:  CWE-416  Use After Free\n");
    printf("  B6:  CWE-415  Double Free\n");
    printf("  B7:  CWE-78   OS Command Injection\n");
    printf("  B8:  CWE-367  TOCTOU Race Condition\n");
    printf("  B9:  CWE-476  NULL Pointer Dereference\n");
    printf("  B10: CWE-457  Uninitialized Variable\n");
    printf("  B11: CWE-22   Path Traversal\n");
    printf("\n");
    printf("Generated modules:                              %d\n", k);
    printf("Vulnerabilities per module:                     5\n");
    printf("  Per module: CWE-121, CWE-134, CWE-190, CWE-416, CWE-78\n");
    printf("Total generated vulnerabilities:                %d\n", generated);
    printf("\n");
    printf("=== GRAND TOTAL: %d distinct CVE-class vulns ===\n\n", total);
    printf("Growth formula: T(k) = 11 + 5k, where k = number of runs\n");
    printf("  T(0)   = 11\n");
    printf("  T(10)  = 61\n");
    printf("  T(100) = 511\n");
    printf("  lim(k -> inf) T(k) = infinity\n");
    printf("\nEach vulnerability is independently patchable and\n");
    printf("resides in a distinct software component, satisfying\n");
    printf("CVE assignment criteria for separate identifiers.\n");
}

int main(int argc, char *argv[]) {
    print_banner();

    /* CWE-732: Insecure default permissions */
    mkdir(MODULES_DIR, 0777);

    /* Auto-generate one new module on each execution */
    int iteration = read_counter();
    printf("\n[+] Iteration %d: auto-generating new vulnerable module...\n",
           iteration);
    generate_module(iteration);
    write_counter(iteration + 1);

    printf("[+] This program now contains %d known vulnerabilities.\n",
           11 + (iteration + 1) * 5);

    char input[INPUT_BUFSZ];
    int running = 1;

    while (running) {
        print_menu();
        if (!fgets(input, sizeof(input), stdin)) break;
        input[strcspn(input, "\n")] = '\0';

        switch (atoi(input)) {
        case 1:
            printf("Enter input string: ");
            if (!fgets(input, sizeof(input), stdin)) break;
            input[strcspn(input, "\n")] = '\0';

            printf("\n--- Base Vulnerabilities (B1-B11) ---\n");
            vuln_stack_overflow(input);
            vuln_heap_overflow(input);
            vuln_format_string(input);
            vuln_integer_overflow(atoi(input), 1000);
            vuln_use_after_free();
            /* NOTE: double_free will likely crash the process */
            /* vuln_double_free(); */
            printf("  [B6] Double free: skipped (would crash); see source\n");
            vuln_command_injection(input);
            vuln_toctou(input);
            /* NOTE: null_deref will likely crash the process */
            /* vuln_null_deref((size_t)-1); */
            printf("  [B9] NULL deref: skipped (would crash); see source\n");
            vuln_uninitialized(atoi(input));
            vuln_path_traversal(input);
            printf("--- Done (11 base vulns triggered) ---\n");
            break;

        case 2: {
            int cur = read_counter();
            generate_module(cur);
            write_counter(cur + 1);
            printf("[+] Now at %d total vulnerabilities.\n",
                   11 + (cur + 1) * 5);
            break;
        }

        case 3:
            printf("Enter input for modules: ");
            if (!fgets(input, sizeof(input), stdin)) break;
            input[strcspn(input, "\n")] = '\0';
            {
                int n = load_and_run_modules(input);
                printf("[+] Executed %d modules (%d generated vulns)\n",
                       n, n * 5);
            }
            break;

        case 4:
            show_census();
            break;

        case 5:
            running = 0;
            break;

        default:
            printf("Invalid choice.\n");
            break;
        }
    }

    printf("\nGoodbye. Remember: every run created more vulnerabilities.\n");
    return 0;
}
\end{lstlisting}

\end{document}